\documentclass[letterpaper, 10 pt, conference]{ieeeconf}  

\IEEEoverridecommandlockouts                              
\usepackage{graphicx}
\usepackage{amsmath} 
\usepackage{amssymb}  
\usepackage{dsfont}

\newtheorem{theo}{Theorem}
\newtheorem{lemma}{Lemma}

\newtheorem{hyp}{Assumption}
\newtheorem{example}{Example}

\DeclareMathOperator*{\argmin}{argmin}

\DeclareMathOperator*{\zer}{zer}
\DeclareMathOperator*{\fix}{fix}

\newcommand\oT{\mathsf T}
\newcommand\oU{\mathsf{U}}

\newcommand\oS{\mathsf{S}}
\newcommand\oR{\mathsf{R}}
\newcommand\oI{\mathsf{I}}
\newcommand\oJ{\mathsf{J}}

\newcommand\oY{\mathsf{Y}}
\newcommand\RR{\mathds{R}} 
\newcommand\NN{\mathds{N}} 
\newcommand\PP{\mathbb{P}} 
\newcommand\EE{\mathbb{E}} 

\newcommand{\leftnorm}{\left|\!\left|\!\left|}
\newcommand{\rightnorm}{\right|\!\right|\!\right|}

\title{Asynchronous Distributed Optimization \\ 
using a Randomized Alternating Direction Method of Multipliers}

\author{Franck Iutzeler, Pascal Bianchi, Philippe Ciblat and Walid Hachem
\thanks{This work was partially funded by the French Defense Agency (DGA).}
\thanks{The authors are with Telecom ParisTech, CNRS LTCI, 75013 Paris, France  {\tt\small \{lastname\}@telecom-paristech.fr}} }

\begin{document}

\setlength{\pdfpageheight}{\paperheight}
\setlength{\pdfpagewidth}{\paperwidth}

\maketitle
\thispagestyle{empty}
\pagestyle{empty}

\begin{abstract}
Consider a set of networked agents endowed with private cost functions         
and seeking to find a consensus on the minimizer of the aggregate cost.        
A new class of random asynchronous distributed optimization methods is         
introduced. The methods generalize the standard Alternating Direction Method
of Multipliers (ADMM) to an asynchronous setting where isolated components 
of the network are activated in an uncoordinated fashion. The algorithms       
rely on the introduction of \emph{randomized} Gauss-Seidel iterations of       
a Douglas-Rachford operator for finding zeros of a sum of two monotone         
operators. Convergence to the sought minimizers is provided under mild         
connectivity conditions. Numerical results sustain our claims.                 
\end{abstract}

\section{Introduction}

Consider a network represented by a set $V$ of agents seeking to solve the following optimization problem
on a Euclidean space $\mathsf X$:
\begin{equation}
\label{eq:pb}
\inf_{x\in\mathsf X} \sum_{v\in V} f_v(x)\ ,
\end{equation}
where $f_v$ is a convex real function known by agent $v$ only.
Function $f_v$ can be interpreted as the price payed by an agent $v$ when the global network state is equal to $x$.

This problem arises for instance in \emph{cloud learning} applications
where massive data sets are distributed in a network and processed by
distinct virtual machines \cite{Forero2011}.  We investigate distributed optimization
algorithms: agents iteratively update a local estimate using their
private objective $f_v$ and, simultaneously, exchange information with
their neighbors in order to eventually reach a consensus on the global
solution.  Standard algorithms are generally \emph{synchronous}: all
agents are supposed to complete their local computations synchronously
at each tick of an external clock, and then synchronously merge
their local results.  However, in many situations, one faces variable
sizes of the local data sets along with heterogeneous computational
abilities of the virtual machines. Synchronism then becomes a burden,
as the global convergence rate is expected to depend on the local computation times of the slowest agents. 
It is crucial to introduce asynchronous methods
which allow the estimates to be updated in a non-coordinated fashion,
rather than all together or in some frozen order.  

The literature contains at least three classes of distributed optimization methods for solving~(\ref{eq:pb}).
The first one is based on the simultaneous use of a local \emph{first-order} optimization algorithm 
(subgradient algorithm~\cite{Bertsekas1989,Sundhar Ram2010,Bianchi2011},  
Nesterov-like method \cite{Jakovetic2012,Duchi2012})
and a gossip process which drives the network to a consensus. 
A second class of methods is formed by distributed Newton-Raphson methods \cite{jadbabaie2009distributed}.
This paper focuses on a third class of methods derived from proximal splitting methods~\cite{Lions1979,Eckstein1992,combettes-pesquet-book1-2011}.
Perhaps the most emblematic proximal splitting method is the so-called \emph{Alternating Direction Method of Multipliers} (ADMM)
recently popularized to multiagent systems by the monograph~\cite{Boyd2011}.
Schizas \emph{et al.} demonstrated the remarkable potential of ADMM to handle distributed optimization problems and introduce a useful 
framework to encompass graph-constrained communications \cite{Schizas2008}.
We also refer to \cite{Mota2012,Wei2012} for recent contributions.
However, all of these works share a common perspective: Algorithms are synchronous. They require a significant amount of coordination or scheduling between agents.
In \cite{Schizas2008,Mota2012}, agents operate  in parallel, whereas \cite{Wei2012} proposes a sequential version of ADMM
where agents operate one after the other in a predetermined order.

\noindent {\bf Contributions.}  This paper introduces a novel class of distributed algorithms to solve~(\ref{eq:pb}).
The algorithms are asynchronous in the sense that some components of the network are allowed to wake up at random and perform local updates, while the rest of the network stands still. No coordinator or global clock is needed. The frequency of activation of the various network components is likely to vary.
The algorithms rely on the introduction of \emph{randomized} Gauss-Seidel iterations of a Douglas-Rachford monotone operator.
We prove that the latter iterations provides a new powerful method for finding the zeros of a sum of two monotone operators.
Application of our method to problem~(\ref{eq:pb}) yields a randomized ADMM-like algorithm, which is proved to converge to the sought minimizers.
\smallskip

The paper is organized as follows. The distributed optimization problem is
rigorously stated in Section \ref{sec:admm}. The synchronous ADMM algorithm
that solves this problem is then described in Section~\ref{sec:sync-admm}. 
Section~\ref{sec:randprox} forms the core of the paper. After quickly 
recalling the monotone operator formalism, the random Gauss-Seidel form 
of the proximal algorithm is described and its convergence is shown there. 
These results will eventually lead to an asynchronous version of the 
well-known Douglas-Rachford splitting algorithm. 
In Section~\ref{sec:randadmm}, the results of Section~\ref{sec:randprox}
are applied towards developing an asynchronous version of the ADMM
algorithm. An implementation example is finally provided in 
Section~\ref{sec:num} along with some simulations in Section~\ref{sec:simus}. 

\subsection*{Notations}

Consider a non-directed graph $G=(V,E)$ where $V$ is a set of vertices 
and $E$ a set of edges. We sometimes note $v\sim w$ for $\{v,w\}\in E$.
For any $A\subset V$, we denote by $G(A)$ the subgraph of $G$ induced by $A$ 
(\emph{i.e.}, $G(A)$ has vertices $A$ and for any $(v,w)\in A^2$, $\{v,w\}$ is an edge of $G(A)$ if and only if 
it is an edge of $G$).
Let $\mathsf X$ be a Euclidean space. 
We denote by $\mathsf X^A$ the set of functions on $A\to\mathsf X$.
It  is endowed with the inner product $\langle x,y\rangle_A = \sum_{v\in A}\langle x(v),y(v)\rangle_{\mathsf X}$
where $\langle \,.\,,\,.\,\rangle_{\mathsf X}$ is the inner product on~$\mathsf X$.
We will omit subscripts $_{\mathsf X}$ and $_A$ when no confusion occurs. 
For any finite collection $A_1,\cdots, A_L\subset V$, we endow the space
$\mathsf X^{A_1}\times\cdots\times\mathsf X^{A_L}$ with the scalar product
$\langle x,y\rangle = \sum_{\ell=1}^L\langle x_\ell,y_\ell\rangle_{A_\ell}$ for any 
$x=(x_1,\cdots,x_L)$ and $y=(y_1,\cdots,y_L)$.

We denote by $\Pi_Ax$ the restriction of $x$ to $A$ \emph{i.e.}, $\Pi_A:\mathsf X^V\to\mathsf X^A$ is the linear operator defined for any $x\in \mathsf X^V$ as
$\Pi_Ax : (v\in A)\mapsto x(v)$. 
We denote by $1_A\in \mathsf X^A$ the constant function equal to one and  by $\text{sp}(1_{A})$ the linear span of $1_A$ \emph{i.e.}, the set of
constant functions on $A$. Notation $|A|$ represents the cardinal of a set $A$.

For a closed proper convex function $h:\mathsf X\to (-\infty,+\infty]$ we define
$\text{prox}_{h,\rho}(x) = \arg\min_y h(y)+\frac\rho 2\|y-x\|^2$.

\section{Distributed Optimization on a Graph}
\label{sec:admm}

Consider a network of agents represented by a non-oriented graph $G=(V,E)$ where $V$ is a finite set of vertices
(\emph{i.e.}, the agents) and $E$ is a set of edges. Each agent $v\in V$ has a private cost function 
$f_v: \mathsf X\to  (-\infty,+\infty]$ where $\mathsf X$ is a Euclidean space.
We make the following assumption on functions $f_v$.
\begin{hyp}
\label{hyp:f}$ $\\
 {\it i)} For all $v\in V$, $f_v$ is a proper closed convex function.\\
 {\it ii)} The infimum in~(\ref{eq:pb}) is finite and is attained at some point~$x^*\in\mathsf X$.
\end{hyp}
In order to solve the optimization problem~(\ref{eq:pb}) on the graph $G$, we first provide
an equivalent formulation of~(\ref{eq:pb}) that will be revealed useful.
For some integer $L\geq 1$, consider a finite collection $A_1, A_2,\cdots, A_L$ of subsets of $V$
which we shall refer to as \emph{components}. 
We assume the following condition.
\begin{hyp}
  {\it i)} $\bigcup_{\ell=1}^L A_\ell = V$.\\
  {\it ii)} $\bigcup_{\ell=1}^L G(A_\ell)$ is connected.
\label{hyp:subg}
\end{hyp}
Assumption~\ref{hyp:subg}{\sl i)} implies that any vertex appears in one of the components $A_1,\cdots, A_L$ at least.
We stress the fact that two distinct components $A_\ell$ and $A_{\ell'}$ are not necessarily disjoint, though. 
Assumption~\ref{hyp:subg}{\sl ii)} means that the union of all subgraphs is connected. As the latter union is also a subgraph of $G$,
this implies that $G$ is connected.
As will be made clear below, our algorithms shall assume that all agents in the same component are able to perform simple operations
in a coordinated fashion (\emph{i.e.}, compute a local average over a component). 
Thus, in practice, it is reasonable to require that each subgraph $G(A_\ell)$ is itself connected.

We introduce some notations. We set for any $x\in \mathsf X^V$,
$$
f(x) \triangleq \sum_{v\in V}f_v(x(v))\ .
$$
For any $z=(z_1,\cdots,z_L)\in \mathsf Z\triangleq \mathsf X^{A_1}\times\cdots\times \mathsf X^{A_L}$, we define the closed proper convex function 
$$
g(z) \triangleq \sum_{\ell=1}^L \iota_{\text{sp}(1_{A_\ell})}(z_\ell)
$$
where $\iota_H$ is the indicator function of a set $H$ (equal to zero on $H$ and to $+\infty$ outside).
Here $g(z)$ is equal to zero if for any $\ell$,  $z_\ell$ is constant. Otherwise, $g(z)$ is infinite.
For any $x\in \mathsf X^V$, we define $Mx \triangleq (\Pi_{A_1}x,\cdots, \Pi_{A_L}x)$.
We consider the following optimization problem:
\begin{equation}
\label{eq:pbEquiv}
\underset{x\in\mathsf X^V}{\inf}\   f(x) + g(Mx)
\end{equation}
\begin{lemma}
  Under Assumption~\ref{hyp:subg}, $x$ is a minimizer of~(\ref{eq:pbEquiv}) if and only if
  $x = \bar x 1_V$ where $\bar x\in \mathsf X$ is a minimizer of~(\ref{eq:pb}). 
\end{lemma}
\begin{proof}
Let $x\in\mathsf X^V$ such that $g(Mx) $ is finite. Then $x$ is constant on each component. 
Let $v,w$ be two arbitrary vertices in $V$. There exists a path in $\bigcup_{\ell=1}^L G(A_\ell)$ connecting $v$ and $w$.
Each edge of this path connects two vertices which belong to a common component. Thus, $x$ is constant on two consecutive vertices of the path.
This proves that $x(v)=x(w)$. Thus, $x$ is constant and the result follows.
\end{proof}
As noted in~\cite{Eckstein1992}, solving Problem~(\ref{eq:pbEquiv}) is         
equivalent to the search of the zeros of two monotone operators.  
One of possible approaches for that sake is to use ADMM. Although 
the choice of the sets $A_1,\cdots, A_L$ does not change the minimizers of the 
initial problem, it has an impact on the particular form of ADMM used to find  
these minimizers, as we shall see below.

In order to be more explicit, we provide in this section two important
examples of possible choices for the components $A_1,\cdots, A_L$. 
\begin{example} Let $L=1$ and
  $A_1=V$. Problem~(\ref{eq:pbEquiv}) writes
  \begin{equation*}
    \underset{x\in\mathsf X^V}{\inf}\   f(x) + \iota_{\text{sp}(1_V)}(x)\ ,
  \end{equation*}
  In this case, the formulation is identical to 
  \cite[Chapter 7]{Boyd2011}. 
\end{example}

\begin{example}
\label{ex:pairwise} 
  Let $L=|E|$ and $\{A_1,\cdots, A_L\}=E$. That is, each set $A_\ell$
  is a pair of vertices $\{v,w\}$ such that $\{v,w\}$ is an
  edge. Problem~(\ref{eq:pbEquiv}) writes
$$
\underset{x\in\mathsf X^V}{\inf}\ f(x) + \sum_{v\sim
  w}\iota_{\text{sp}(\boldsymbol 1_2)}\!\left(
  \begin{array}[h]{@{}c@{}}
    x(v) \\ x(w)
  \end{array}\right)
$$
where $\boldsymbol 1_2$ stands for the vector $(1,1)^T$.
\end{example}

\section{Synchronous ADMM}
\label{sec:sync-admm}

\subsection{General facts}

We now apply the standard ADMM to Problem~(\ref{eq:pbEquiv}).
Perhaps the most direct way to describe ADMM is to reformulate the unconstrained problem~(\ref{eq:pbEquiv})
into the following constrained problem: Minimize $f(x)+g(z)$ subject to $z=Mx$.
For any $x\in \mathsf X^V$, $\lambda,z\in\mathsf Z$, the augmented Lagrangian is given by
\begin{equation}
\label{eq:alag}
\mathcal{L}_\rho(x,z;\lambda) \triangleq \hspace*{-1mm} f(x) + g(z) + \langle \lambda , Mx-z \rangle + \frac{\rho}{2}\left\| Mx - z \right\|^2
\end{equation}
where $\rho>0$ is a constant. ADMM consists of the iterations 
\begin{subequations} \label{eq:LADMM}
\begin{align}
\hspace*{-3mm} x^{k+1} &  = \underset{x\in \mathsf X^V}{\argmin} ~ \mathcal{L}_\rho(x,z^k;\lambda^k)  \label{eq:LADMMx}\\
\hspace*{-3mm} z^{k+1} &  = \underset{z\in\mathsf Z}{\argmin} ~ \mathcal{L}_\rho(x^{k+1},z;\lambda^k) \label{eq:LADMMz} \\
\hspace*{-3mm}\lambda^{k+1} & =   \lambda^{k} + \rho \left( Mx^{k+1} - z^{k+1} \right) . \label{eq:LADMMl}
\end{align}
\end{subequations}
From \cite[Chap. 3.2]{Boyd2011}, the following result is immediate.
\begin{theo} 
Under Assumption~\ref{hyp:f}, the sequence $(x^k)$ defined in (\ref{eq:LADMMx}) 
converges to a minimizer of~\eqref{eq:pbEquiv}.
\end{theo}

\subsection{Decentralized Implementation}

One should now make~(\ref{eq:LADMM}) more explicit 
and convince the reader that the iterations are indeed amenable to distributed implementation.
Due to the specific form of function $g$, it is clear from~(\ref{eq:LADMMz})
that all components $z_1^k,\cdots,z_L^k$ of $z^k$ are constant. Otherwise stated, 
$z^k = (\bar z_1^k1_{A_1},\cdots,\bar z_L^k1_{A_L})$ for some constants $\bar z_\ell^k\in \mathsf X$.
For any $v\in V$, we set
$$
\sigma(v) \triangleq \{\ell\,:\,v\in A_\ell\}\ .
$$
Now consider the first update equation~(\ref{eq:LADMMx}). Getting rid of all quantities in $\mathcal L_\rho$
which do not depend on the $v$th component of $x$, we obtain for any $v\in V$
$$
x^{k+1}(v) = \underset{y\in \mathsf X}{\argmin}~ f_v(y) + \sum_{\ell\in\sigma(v)}\langle\lambda_\ell^k(v),y\rangle+\frac\rho 2 \|y-\bar z^k_\ell\|^2\,.
$$
After some algebra, the above equation further simplifies to
\begin{equation}
  \label{eq:proxx}
  x^{k+1}(v) = \text{prox}_{f_v,\rho |\sigma(v)|}\left( 
Z^k(v) - B^k(v)\right)
\end{equation}
where we introduced the following constants:
\begin{eqnarray}
\label{eq:Z} Z^k(v) = \frac 1{|\sigma(v)|}\sum_{\ell\in\sigma(v)}\hspace*{-2mm} \bar z^k_\ell, B^k(v) = \frac 1{\rho|\sigma(v)|}\sum_{\ell\in\sigma(v)} \hspace*{-2mm} \lambda^k_\ell(v)\,.
\end{eqnarray}
It is straightforward to show that the second update equation (\ref{eq:LADMMz}) admits as well a simple decomposable form.
After some algebra, we obtain that for any $\ell=1,\cdots,L$,
\begin{equation}
  \label{eq:2}
  \bar z_\ell^{k+1} = \frac 1{|A_\ell|}\sum_{v\in A_\ell} x^{k+1}(v) + \frac{\lambda_\ell^k(v)}\rho\,.
\end{equation}
Finally, for all $\ell=1,\cdots,L$ and $v\in A_\ell$, equation~(\ref{eq:LADMMl}) reads  
\begin{equation}
\lambda_\ell^{k+1}(v) = \lambda_\ell^{k}(v)+ \rho(x^{k+1}(v)-\bar z_\ell^{k+1})\,.
\label{eq:4}
\end{equation}
Averaging~(\ref{eq:4})  w.r.t. $v$ and using~(\ref{eq:2}) yields $\sum_{v\in A_\ell}\lambda_\ell^{k}(v)=0$.
Thus, the second term in the RHS of~(\ref{eq:2}) can be deleted. Finally, averaging~(\ref{eq:4}) w.r.t. $\ell$ leads to
\begin{equation}
B^{k+1}(v) = B^{k}(v)+ x^{k+1}(v)-Z^{k+1}(v)\,.
   \label{eq:LADMMlter}
\end{equation}

  \noindent {\bf Synchronous ADMM}:
\noindent \hrulefill
\nopagebreak\\
\noindent At each iteration $k$, \\
For each agent $v$, compute $x^{k+1}(v)$ using~(\ref{eq:proxx}).\\
In each components $\ell=1,\cdots,L$, compute $$\bar z_\ell^{k+1} = \frac 1{|A_\ell|}\sum_{w\in A_\ell} x^{k+1}(w).$$
For each agent $v$, compute $Z^{k+1}(v)$ and $B^{k+1}(v)$ using~(\ref{eq:Z}) and~(\ref{eq:LADMMlter}) respectively.

\noindent \hrulefill
\smallskip

The above algorithm implicitly requires the existence of a routine for computing an average, in each component $A_\ell$.
This requirement is mild when the components coincide with edges of the graph as in Example~2. In this case, one only
needs that the two vertices of an edge share their current estimate and find an agreement on the average.
In the general case, the objective can be achieved by selecting a leader in each component whose role is to gather
the estimates, compute the average and send the result to all agents in this component.

It is worth noting that in the case of Example~1, the synchronous ADMM described above coincides with the algorithm of \cite{Boyd2011}.

\section{A Randomized Proximal Algorithm}
\label{sec:randprox}



\subsection{Monotone operators}
\label{sec:mono}

An operator $\oT$ on a Euclidean space $\oY$ is a set valued mapping 
$\oT : {\oY} \to 2^{\oY}$. 
An operator can be equivalently identified with a subset of 
$\oY \times \oY$, and we write $(x,y)\in\oT$ when
$y\in\oT(x)$. Given two operators $\oT_1$ and $\oT_2$ on $\oY$ and two
real numbers $\alpha_1$ and $\alpha_2$, the operator $\alpha_1 \oT_1 + \alpha_2
\oT_2$ is defined as $\alpha_1 \oT_1 + \alpha_2 \oT_2 = \{ (x, \alpha_1 y_1 +
\alpha_2 y_2) \, : \, (x, y_1) \in \oT_1, \, (x, y_2) \in \oT_2 \}$.  
The identity operator is $\oI = \{ (x,x) : x \in \oY\}$ and the inverse
of the operator $\oT$ is $\oT^{-1} = \{ (x,y) : (y,x)\in \oT \}$. 
The operator $\oT$ is said \emph{monotone} if 
\[
\forall \ (x,y),(x',y') \in \oT, \ 
\langle x - x' , y- y' \rangle \geq 0 . 
\]
A monotone operator is said \emph{maximal} if it is not strictly contained in any 
monotone operator (as a subset of $\oY \times \oY$). 
Finally, $\oT$ is said \emph{firmly non-expansive} if 
\[
\forall \ (x,y),(x',y') \in \oT, \ 
\langle x - x' , y- y' \rangle \geq \| y - y' \|^2. 
\]
The typical example of a monotone operator is the subdifferential 
$\partial f$ of a convex function $f : \oY \to \RR$. Finding a minimum
of $f$ amounts to finding a point in $\zer(\partial f)$, where 
$\zer(\oT) = \{ x : 0 \in \oT(x) \}$ is the set of zeroes of an operator $\oT$. 
A common technique for finding a zero of a maximal monotone operator $\oT$ is 
the so-called \textit{proximal point algorithm} \cite{Rockafellar1976} that
we now describe. 
The \emph{resolvent} of $\oT$ is the operator 
$\oJ_{\rho\oT} \triangleq (\oI + \rho \oT)^{-1} $ for $\rho>0$. 
One key result (see \emph{e.g.} \cite{Eckstein1992}) 
says that $\oT$ is maximal monotone if and only if $\oJ_{\rho\oT}$ is 
firmly non expansive and its domain is $\oY$.  
Observe that a firmly non expansive operator is single valued and denote by
$\fix(\oJ_{\rho\oT})$ the set of fixed points of $\oJ_{\rho\oT}$. It is
clear that $\fix(\oJ_{\rho\oT}) = \zer(\oT)$. 
The firm non expansiveness of $\oJ_{\rho\oT}$ plays a central role in the 
proof of the following result: 
\begin{lemma}[Proximal point algorithm \cite{Rockafellar1976}]
\label{lemma:prox}
If $\oT$ is a maximal monotone operator and $\rho>0$, then the iterates 
$\zeta^{k+1} = \oJ_{\rho\oT} (\zeta^k)$ starting at any point of $\oY$
converge to a point of $\fix(\oJ_{\rho\oT})$ whenever this set is non-empty.
\end{lemma}

\subsection{Random Gauss-Seidel iterations} 
\label{sec:randomgs}

Assume now that the Euclidean space $\oY$ is a Cartesian product of Euclidean
spaces of the form $\oY = \oY_1 \times \cdots \times \oY_L$ where $L$ is a 
given integer, and write any $\zeta \in \oY$ as 
$\zeta = ( \zeta_1, \ldots, \zeta_L )$ where $\zeta_\ell \in \oY_\ell$ for
$\ell=1,\ldots, L$. 
Let $\oS$ be a firmly non expansive operator on $\oY$ and write  
\[
\oS(\zeta) = \left( \oS_1(\zeta) , ... , \oS_L(\zeta)  \right)
\] 
where $\oS_\ell(\zeta) \in \oY_\ell$. For $\ell=1,\ldots,L$, define the single
valued operator $\hat\oS_\ell : \oY \to \oY$ as 
\begin{equation}
\label{hat-S} 
\hat\oS_\ell(\zeta) = \left( \zeta_1, \ldots, \zeta_{\ell-1}, 
\oS_\ell(\zeta), \zeta_{\ell+1},\ldots, \zeta_L \right) .
\end{equation} 
Considering an iterative algorithm of the form $\zeta^{k+1} = \oS (\zeta^k)$, 
its \emph{Gauss-Seidel} version would be an algorithm of the form
$\zeta^{k+1} = \hat{\oS}_L \circ \cdots \circ \hat{\oS}_1(\zeta^k)$. We are
interested here in a \emph{randomized version} of these iterates. On a 
probability space $(\Omega, {\mathcal F}, \PP)$, let  
$(\xi^k)_{k\in \NN}$ be a random process satisfying the 
following assumption: 
\begin{hyp}
\label{hyp:xi} 
The random variables $\xi^k$ are independent and identically distributed. 
They are valued in the set $\{ 1,\ldots, L\}$ with 
$\PP[\xi^1 = \ell] = p_\ell > 0$ for all $\ell=1,\ldots, L$.
\end{hyp} 
We are interested here in the convergence  
of the random iterates $\zeta^{k+1} = \hat\oS_{\xi^{k+1}}(\zeta^k)$ towards
a (generally random) point of $\fix(\oS)$, provided this set is non empty: 

\begin{theo}[Main result]
\label{theo:main}
Let $\oS$ is a firmly non-expansive operator on $\oY$ with domain $\oY$. 
Let $(\xi^k)_{k\in\NN}$ be a sequence of random variables satisfying
Assumption \ref{hyp:xi}. Assume that $\fix(\oS) \neq \emptyset$. Then for 
any initial value $\zeta^0$, the sequence of iterates 
$\zeta^{k+1} = \hat{\oS}_{\xi^{k+1}} (\zeta^k)$ converges almost surely to
a random variable supported by $\fix(\oS)$. 
\end{theo}

\begin{proof}
Denote by $\langle \zeta, \eta \rangle = \sum_{\ell=1}^L 
\langle \zeta_\ell, \eta_\ell \rangle_{\oY_\ell}$ the inner product of
$\oY$, and by $\| \zeta \|^2 = \langle \zeta, \zeta \rangle$
its associated squared norm. Define a new inner product
$\zeta \bullet \eta = \sum_{\ell=1}^K p_\ell^{-1} 
\langle \zeta_\ell, \eta_\ell \rangle_{\oY_\ell}$ on $\oY$, and let 
$\leftnorm \zeta \rightnorm^2 = \zeta \bullet\zeta$ 
be its associated squared norm. 
Fix $\zeta^\star$ in $\fix(\oS)$. Conditionally
to the sigma-field $\mathcal{F}_k = \sigma(\xi^1,\ldots, \xi^k)$ we have
\begin{align*}
 & \EE[\leftnorm\zeta^{k+1} - \zeta^\star\rightnorm^2 \, | \, 
 \mathcal{F}_k ] 
= \sum_{\ell=1}^L p_\ell 
     \leftnorm \hat{\oS}_{\ell} (\zeta^k) -  \zeta^\star\rightnorm^2 \\
 & = \sum_{\ell=1}^L p_\ell  \Bigl(  
\frac{1}{p_\ell}  \| \oS_{\ell}(\zeta^k) - \zeta_\ell^\star\|^2_{\oY_\ell} +  
\sum_{\underset{i \neq \ell}{i=1}}^L  \frac{1}{p_i}  
\| \zeta_i^k -  \zeta_i^\star\|^2_{\oY_i} \Bigr)   \\
 & =  \| \oS(\zeta^k) - \zeta^\star\|^2 +  
\sum_{\ell=1}^L  \frac{1-p_\ell}{p_\ell}  
\| \zeta_\ell^k - \zeta_\ell^\star\|^2_{\oY_\ell}   \\
 & =   \leftnorm \zeta^k -  \zeta^\star\rightnorm^2 +  
\| \oS(\zeta^k) - \zeta^\star\|^2 - \| \zeta^k -  \zeta^\star\|^2   
\end{align*}
Since $(\oI - \oS)(\zeta^\star) = 0$, we have 
\begin{align}
\nonumber & \| \oS(\zeta^k) - \zeta^\star\|^2 - \| \zeta^k -  \zeta^\star\|^2  \\
\nonumber & = \| \oS(\zeta^k) - \zeta^k + \zeta^k - \zeta^\star\|^2 - 
                                       \| \zeta^k -  \zeta^\star\|^2  \\
\nonumber & =  \| \oS(\zeta^k) - \zeta^k\|^2  
+ 2 \langle \oS(\zeta^k) - \zeta^k  , \zeta^k -  \zeta^\star \rangle \\
\nonumber & =  \| \oS(\zeta^k) - \zeta^k\|^2  
 - 2 \langle  (\oI - \oS)(\zeta^k) - (\oI - \oS)(\zeta^\star) , 
                                     \zeta^k -  \zeta^\star \rangle \\
\nonumber  & \leq  -  \| \oS(\zeta^k) - \zeta^k\|^2 
\end{align}
where the inequality comes from the easily verifiable fact that 
$(\oI - \oS)$ is firmly non-expansive when $\oS$ is. This leads to the 
inequality 
\begin{equation}
\label{surmartingale} 
 \EE[\leftnorm\zeta^{k+1} - \zeta^\star\rightnorm^2  | \mathcal{F}_k ] 
  \leq  \leftnorm \zeta^k -  \zeta^\star\rightnorm^2 -  
                                  \| \oS(\zeta^k) - \zeta^k\|^2 
\end{equation} 
which shows that $\leftnorm\zeta^{k} - \zeta^\star\rightnorm^2$ is a 
nonnegative supermartingale with respect to the filtration $({\mathcal F}_k)$. 
As such, it converges with probability one towards a random variable 
$X_{\zeta^\star}$ satisfying $0\leq X_{\zeta^\star} < \infty$ almost
everywhere. 
Given a countable dense subset $H$ of $\fix(\oS)$, there is a probability
one set on which $\leftnorm \zeta^k - {\boldsymbol \zeta} 
\rightnorm \to X_{\boldsymbol\zeta} \in [0, \infty)$ for all 
${\boldsymbol\zeta} \in H$. 
Let $\zeta^\star \in \fix(\oS)$, let $\varepsilon > 0$, and choose 
${\boldsymbol\zeta} \in H$ such that $\leftnorm \zeta^\star - \boldsymbol\zeta
\rightnorm \leq \varepsilon$. With probability one, we have  
\[
\leftnorm \zeta^k - \zeta^\star \rightnorm \leq  
\leftnorm \zeta^k - \boldsymbol\zeta \rightnorm + 
\leftnorm \boldsymbol\zeta - \zeta^\star \rightnorm \leq 
X_{\boldsymbol\zeta} + 2 \varepsilon 
\]  
for $k$ large enough. Similarly, 
$\leftnorm \zeta^k - \zeta^\star \rightnorm \geq  
X_{\boldsymbol\zeta} - 2 \varepsilon$ for $k$ large enough.  
We therefore obtain: 
\begin{description}
\item[{\bf C1 :}] There is a probability one set on which
$\leftnorm \zeta^k - \zeta^\star \rightnorm$ converges for every 
$\zeta^\star \in \fix(\oS)$. 
\end{description} 
Getting back to Inequality \eqref{surmartingale}, taking the expectations on
both sides of this inequality and iterating over $k$, we obtain 
\[
\sum_{k=0}^\infty \EE [ \| \oS(\zeta^k) - \zeta^k\|^2 ] \leq 
(\zeta^0 - \zeta^\star)^2 . 
\]
By Markov's inequality and Borel Cantelli's lemma, we therefore obtain: 
\begin{description}
\item[{\bf C2 :}] $S(\zeta^k) - \zeta^k \to 0$ almost surely. 
\end{description} 
We now consider an elementary event in the probability one set where {\bf C1} 
and {\bf C2} hold. On this event, 
since $\leftnorm \zeta^k - \zeta^\star \rightnorm$ converges for 
$\zeta^\star \in \fix(\oS)$, the sequence $\zeta^k$ is bounded. 
Since $\oS$ is firmly non expansive, it is continuous, and {\bf C2} shows that 
all the accumulation points of $\zeta^k$ are in $\fix(\oS)$. It remains to show
that these accumulation points reduce to one point. Assume that 
$\zeta_1^\star$ is an accumulation point. By {\bf C1}, 
$\leftnorm \zeta^k - \zeta^\star_1 \rightnorm$ converges. Therefore, 
$\lim \leftnorm \zeta^k - \zeta^\star_1 \rightnorm = 
\liminf \leftnorm \zeta^k - \zeta^\star_1 \rightnorm = 0$, which shows that
$\zeta^\star_1$ is unique. 
\end{proof}

\section{Random ADMM}
\label{sec:randadmm} 

We now return to the optimization problem~(\ref{eq:pbEquiv}).
It is a well known fact that the standard ADMM can be seen as special case of the so-called Douglas-Rachford algorithm~\cite{Eckstein1992}.
The Douglas-Rachford algorithm can itself be seen as a special case of a proximal point algorithm.
By the results of the previous section, this suggests that random Gauss-Seidel iterations applied to the Douglas-Rachford operator
produce a sequence which eventually converges to the sought solutions.
It turns out that the latter random iterations can be written under the form of practical asynchronous ADMM-like algorithm.

\subsection{Douglas-Rachford operator}

Consider the following dual problem associated with~(\ref{eq:pbEquiv}) 
\begin{equation}
\label{eq:dual}
\min_{\lambda\in \mathsf Z} f^*(-M^*\lambda) + g^*(\lambda)\,,
\end{equation}
where $f^*, g^*$ are the Fenchel conjugates of $f$ and $g$ and $M^*$ is the adjoint of $M$.
By Assumption~\ref{hyp:f} along with~\cite[Th.3.3.5]{Borwein2006}, the 
minimum in~(\ref{eq:dual}) is attained and its opposite coincides with the 
minimum of~(\ref{eq:pbEquiv}).
Note that $\lambda$ is a minimizer of~(\ref{eq:dual}) iff 
zero belongs to the subdifferential of the objective function in~(\ref{eq:dual}). 
By \cite[Th.3.3.5]{Borwein2006} again, this reads
$0\in-M\cdot\partial f^*(-M^*\lambda) + \partial g^*(\lambda)$.
Otherwise stated, finding minimizers of the dual problem~(\ref{eq:dual}) boils
down to searching zeros of the sum of two maximal monotone operators $\oT+\oU$
defined by $\oT = -M\cdot\partial f^*\circ (-M^*)$ and $\oU = \partial g^*$. 
For a fixed $\rho>0$, the Douglas-Rachford~/~Lions-Mercier operator $\oR$ is
defined as 
\begin{equation*}
\left\{ (\nu + \rho b, \mu-\nu) : (\mu,b)\in \oU , (\nu,a)\in\oT , \nu + \rho a = \mu - \rho b \right\}.
\end{equation*}
The following Lemma is an immediate consequence of \cite{Eckstein1992}.
\begin{lemma}
\label{lemma:LM}
Under Assumption~\ref{hyp:f}, $\oR$ is maximal monotone, and 
$\zer(\oR) \neq \emptyset$. 
Moreover, $\oJ_{\rho\oU}(\zeta)\in \zer(\oT+\oU)$ for 
any $\zeta\in\zer(\oR)$. 
\end{lemma}
Lemma~\ref{lemma:LM} implies that the search for a zero of $\oT+\oU$ boils down to the search of a zero of $\oR$ up to a resolvent step $\oJ_{\rho \oU}$.
To that end, a standard approach is to use a proximal point algorithm of the form $\zeta^{k+1} = \oJ_{\oR}(\zeta^k)$.
By~\cite{Eckstein1992}, it can be shown that this approach is equivalent to the ADMM derived in Section~\ref{sec:admm}.
Here, our aim is different. We shall consider random Gauss-Seidel iterations in order to derive
an asynchronous version of the ADMM.

\subsection{Random Gauss-Seidel Iterations}

Define $\oS \triangleq \oJ_\oR$ as the resolvent associated with the Douglas-Rachford operator $\oR$.
On the space $\mathsf Z = \mathsf X^{A_1}\times\cdots\times \mathsf X^{A_L}$, 
define the operator $\hat\oS_\ell$ as in \eqref{hat-S} for any 
$\ell=1,\cdots,L$. 
Let  $(\xi^k)_{k\in\NN}$ be a random process satisfying 
Assumption~\ref{hyp:xi}. 
The following result is a consequence of Theorem~\ref{theo:main} combined with Lemma~\ref{lemma:LM}. 
\begin{theo}
Let Assumptions~\ref{hyp:f}, \ref{hyp:subg} and \ref{hyp:xi} hold true.
Consider the sequence $(\zeta^k)_k$ defined by $\zeta^{k+1} = \hat \oS_{\xi^{k+1}}(\zeta^k)$.
Then for any initial value $\zeta^0$, the sequence $\lambda^k \triangleq \oJ_{\rho\oU}(\zeta^k)$ converges almost surely to a minimizer of~(\ref{eq:dual}).
\end{theo}
In order to complete the above result, we still must justify the fact that, as claimed, the above iterations can be seen as 
an asynchronous distributed algorithm. 

\subsection{Distributed Algorithm}

We make the above random Gauss-Seidel iterations more explicit. In the sequel we shall always denote by $\zeta_\ell$ the $\ell$th component of a 
function $\zeta\in \mathsf Z$ \emph{i.e.}, $\zeta = (\zeta_1,\cdots,\zeta_L)$. 
For any $\ell$, we introduce the average $\bar \zeta_\ell = \sum_{v\in A_\ell} \zeta_\ell(v) / |A_\ell|$.
Lemma~\ref{lem:Sl} below states that any $\zeta\in \mathsf Z$ is uniquely represented by a couple $(\lambda,z)\in \oU$ whose expression is provided.
Moreover, it provides the explicit form of the $\ell$th block $\oS_\ell$ of the resolvent $\oS$. This shall be the basis
of our asynchronous distributed algorithm.
\begin{lemma}
For any $\zeta\in \mathsf Z$, the following holds true.\\
{\it i)} There exist a unique $(\lambda,z)\in \oU$ such that $\lambda+\rho z = \zeta$. \\
{\it ii)} $\oJ_{\rho U}(\zeta)=\lambda$.\\
{\it iii)} For any $\ell=1,\cdots, L$,  $\lambda_\ell = \zeta_\ell - \bar\zeta_\ell 1_{A_\ell}$ and $z_\ell = \frac{\bar\zeta_\ell}\rho 1_{A_\ell}.$\\
{\it iv)} For any $\ell=1,\cdots, L$, and any $v\in A_\ell$
\begin{equation}
\label{eq:Slzeta}
\oS_\ell(\zeta) : v\mapsto \lambda_\ell(v) + \rho x(v)
\end{equation}
where $x(v)$ is defined by
\begin{equation}
\label{eq:Slx}
x(v) = \text{prox}_{f_v,\rho|\sigma(v)|}\left(\frac 1{|\sigma(v)|}\sum_{\ell\in\sigma(v)}\bar z_\ell - \frac{\lambda_\ell(v)}\rho\right)\,.
\end{equation}
\label{lem:Sl}
\end{lemma}
\begin{proof}
{\it i)-ii)} {\it Existence}: Let us define $\lambda = \oJ_{\rho \oU}(\zeta)$ and $z = (\zeta- \lambda)/\rho$. Trivially, $\lambda+\rho z = \zeta$. 
As $\zeta\in \lambda + \rho \oU(\lambda)$, we deduce that $(\lambda,z)\in\oU$.
{\it Uniqueness}: For a fixed $(\lambda, z)\in \oU$ satisfying  $\lambda+\rho z = \zeta$, one has $\zeta\in (I+\rho\oU)(\lambda)$ and thus
$\lambda = \oJ_{\rho \oU}(\zeta)$. As a consequence, $z = (\zeta- \lambda)/\rho$.

{\it iii)} We use $\lambda = \oJ_{\rho \oU}(\zeta) = \text{prox}_{g^*,\rho}(\zeta) = \zeta - \text{prox}_{g,\rho}(\zeta)$
(see~\cite[Th. 14.3]{Bauschke2011}).
As $g$ is the indicator function of the set $\text{sp}(1_{A_1})\times\cdots\times\text{sp}(1_{A_L})$, $\text{prox}_{g,\rho}$ coincides with the projection
operator onto that set. Thus, for any $\ell$, $\lambda_\ell = \zeta_\ell - \bar\zeta_\ell 1_{A_\ell}$.
The expression of $z$ follows from  $z = (\zeta- \lambda)/\rho$. 

{\it iv)} Operator $\oS=\oJ_\oR$ can be written as
$$
\left\{ (\mu+\rho b,\nu + \rho b) : (\mu,b)\in \oU , (\nu,a)\in\oT , \nu + \rho a = \mu - \rho b \right\}.
$$
Moreover, as $\oR$ is monotone, $\oS(\zeta)$ is a singleton.
Representing $\zeta=\lambda+\rho z$ with $(\lambda,z)\in\oU$, it follows from the above expression of $\oS$ that
$\oS(\zeta) = \nu+\rho z$ where $\nu$ is such that $\nu + \rho a = \lambda - \rho z$ for some $a\in\oT(\nu)$.
Using $\oT = -M\cdot\partial f^*\circ (-M^*)$, condition $a\in\oT(\nu)$ translates to: there exists $x\in \partial f^*(M^*\nu)$
s.t. $a=-Mx$. The output-resolvent is obtained by $\nu+\rho z = \lambda + \rho Mx$. 
For a given component $\ell$, this boils down to equation~(\ref{eq:Slzeta}).
The remaining task is to provide the expression of $x$.
By the Fenchel-Young equality $\partial f^* = \partial f^{-1}$~\cite[Prop.3.3.4]{Borwein2006}, 
condition $x\in \partial f^*(M^*\nu)$ is equivalent to $M^*\nu\in\partial f(x)$. 
Using that $\nu = \lambda - \rho (z -Mx)$, we obtain
$0\in \partial f(x) -M^*\lambda +  \rho M^*(z -Mx)$. Otherwise stated,
$x = \arg\min_{y\in\mathsf X^V}{\mathcal L}_\rho(y,z;\lambda)$ where $\mathcal L_\rho$ is the augmented Lagrangian
defined in~(\ref{eq:alag}). 
Using the results of Section~\ref{sec:sync-admm},  $x(v)$ is given by (\ref{eq:Slx}) for any~$v$.
\end{proof}
We are now in position to state the main algorithm. It simply consists
in an explicit writing of the random Gauss-Seidel iterations
$\zeta^{k+1} = \hat \oS_{\xi^{k+1}}(\zeta^k)$ using Lemma~\ref{lem:Sl}{\it iv)}. 
Note that, by Lemma~\ref{lem:Sl}{\it i)}, the definition of a sequence $(\zeta_k)_k$ on $\mathsf Z$
is equivalent to the definition of two sequences $(\lambda^k,z^k)\in \oU$ such that $\zeta^k = \lambda^k + \rho z^k$.
Moreover, by Lemma~\ref{lem:Sl}{\it iii)}, each component $z^k_\ell$ of  $z^k$ is a constant. The definition of $z^k$ thus reduces to the definition
of $L$ constants $\bar z^k_1,\cdots,\bar z^k_L$ in $\mathsf X$.

\noindent {\bf Asynchronous ADMM}:
\noindent \hrulefill
\nopagebreak\\
\noindent At each iteration $k$, draw r.v. $\xi^{k+1}$.\\
For $\ell=\xi^{k+1}$, set for any $v\in A_\ell$:
\begin{eqnarray*}
  x^{k+1}(v) &=& \text{prox}_{f_v,\rho|\sigma(v)|}\left(\frac 1{|\sigma(v)|}\sum_{\ell\in\sigma(v)}\bar z^k_\ell - \frac{\lambda^k_\ell(v)}\rho\right)\\
\bar z_\ell^{k+1} &=& \frac 1{|A_\ell|}\sum_{w\in A_\ell} x^{k+1}(w) \\
\lambda^{k+1}_\ell(v) &=& \lambda^{k}(v) +\rho \left(x^{k+1}(v) - \bar z_\ell^{k+1}\right)\,.
\end{eqnarray*}
For any $\ell\neq \xi^{k+1}$, set $\lambda^{k+1}_\ell = \lambda^{k}_\ell$.\\
For any $w\notin A_{\xi^{k+1}}$, set $x^{k+1}(w)=x^{k}(w)$.

\noindent \hrulefill
\smallskip

\section{Implementation Example}
\label{sec:num}

In order to illustrate our results, we consider herein an asynchronous 
version of the ADMM algorithm in the context of 
Section~\ref{sec:admm}-Example~\ref{ex:pairwise}. 
The scenario is the following: first, Agent $v \in \{ 1, \ldots, | V| \}$ 
wakes up at time $k+1$ with the probability $q_v$. Denoting by 
${\mathcal N}_v$ the neighborhood of Agent $v$ in the Graph $G$, this agent
then chooses one of its neighbors, say $w$, with the probability 
$1/|{\mathcal N}_v|$ and sends an activation message to $w$. In this setting, 
the edge $\{v, w\}$ coincides with one of the 
$A_\ell$ of Example~\ref{ex:pairwise} in Section~\ref{sec:admm}. It is easy
to see that the samples of the activation process $\xi^k$ who is of course
valued in $E$ are governed by the probability law 
\[
\PP[ \xi^1 = \{ v, w \} ] = \frac{q_v}{|{\mathcal N}_v|} 
+ \frac{q_w}{|{\mathcal N}_w|} > 0 . 
\] 
When the edge $\{v,w\}$ is activated, the following two $\text{prox}(\cdot)$
operations are performed by the agents: 
\begin{eqnarray*}
 x^{k+1}(v) &=& \text{prox}_{f_v,\rho|{\mathcal N}(v)|}
\Bigl(\frac 1{|{\mathcal N}(v)|} \sum_{\ell\in{\mathcal N}(v)}
\bar z^k_\ell - \frac{\lambda^k_\ell(v)}\rho\Bigr)\\
 x^{k+1}(w) &=& \text{prox}_{f_w,\rho|{\mathcal N}(w)|}
\Bigl(\frac 1{|{\mathcal N}(w)|} \sum_{\ell\in{\mathcal N}(w)}
\bar z^k_\ell - \frac{\lambda^k_\ell(w)}\rho\Bigr). 
\end{eqnarray*}
The two agents exchange then the values $x^{k+1}(v)$ and $x^{k+1}(w)$ and 
perform the following operations: 
\begin{eqnarray*}
\bar z_\ell^{k+1} &=& \frac{x^{k+1}(v) + x^{k+1}(w)}{2} \\  
\lambda^{k+1}_\ell(v) &=& \lambda^{k}(v) 
                   +\rho \frac{x^{k+1}(v) - x^{k+1}(w)}{2}  \\
\lambda^{k+1}_\ell(w) &=& \lambda^{k}(w) 
                   +\rho \frac{x^{k+1}(w) - x^{k+1}(v)}{2}  . 
\end{eqnarray*}
We remark that this communication scheme is reminiscent of the so-called 
\textit{Random Gossip} algorithm introduced in \cite{Boyd2006} in the context
of distributed averaging. 

\section{Numerical Results}
\label{sec:simus} 

We consider a network with $V = \{1,\ldots, 5\}$ and with 
$E = \{ \{1,2\}, \{ 2,3\}, \{ 3,4\}, \{4,5\}, \{ 5,3 \}\}$.
We evaluate the behavior of: i) the \textit{Synchronous ADMM} ii) the
\textit{Asynchronous ADMM} and iii) the \textit{Distributed Gradient Descent}
with $1/\sqrt{k}$ stepsize \cite{Tsitsiklis1986} using \textit{Random Gossip}
as a communication algorithm\cite{Boyd2006}. Each agent maintains a different
quadratic convex function and their goal is to reach consensus over the
minimizer of problem (\ref{eq:pb}).

In Figure~\ref{fig:cdc}, we plot the squared error versus the number of primal 
updates for the three considered algorithms. We observe that our algorithm     
clearly outperforms the Distributed Gradient Descent. 

\begin{figure}[!ht]
\centering
\includegraphics[width=\columnwidth]{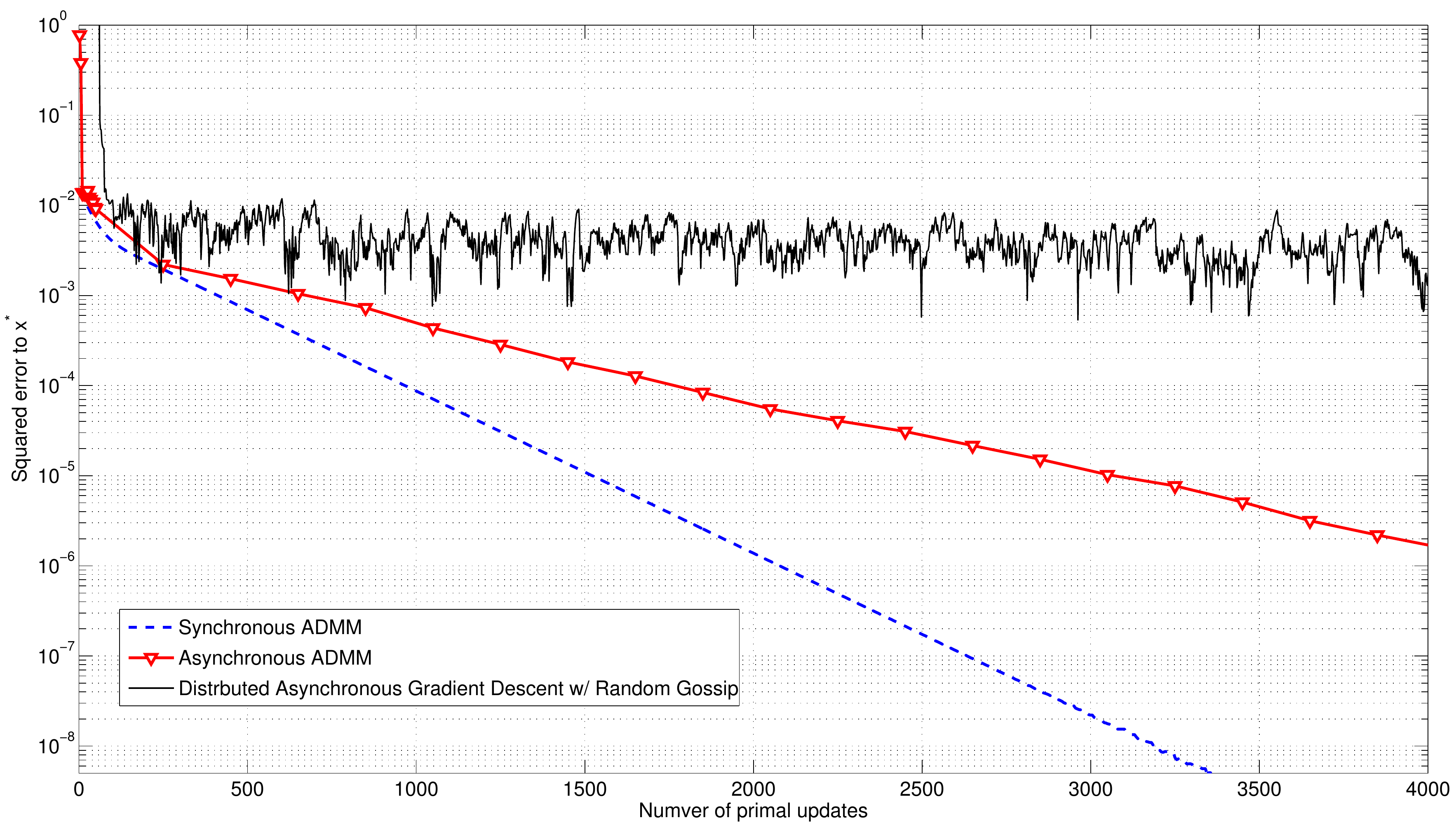}
\caption{Squared error of distributed optimization algorithms versus the number of primal updates}
\label{fig:cdc}
\end{figure}



\tiny



\begin{thebibliography}{10}

\bibitem{Forero2011}
P.~A. Forero, A.~Cano, and G.~B. Giannakis,
\newblock ``{D}istributed {C}lustering {U}sing {W}ireless {S}ensor
  {N}etworks,''
\newblock {\em IEEE Journal of Selected Topics in Signal Processing}, vol. 5,
  no. 4, pp. 707--724, 2011.

\bibitem{Bertsekas1989}
Dimitri~P Bertsekas and John~N Tsitsiklis,
\newblock {\em {P}arallel and distributed computation},
\newblock Old Tappan, NJ (USA); Prentice Hall Inc., 1989.

\bibitem{Sundhar Ram2010}
S.~Ram, A.~Nedic, and V.~Veeravalli,
\newblock ``{D}istributed {S}tochastic {S}ubgradient {P}rojection {A}lgorithms
  for {C}onvex {O}ptimization,''
\newblock {\em Journal of Optimization Theory and Applications}, vol. 147, no.
  3, pp. 516--545, 2010.

\bibitem{Bianchi2011}
P.~Bianchi and J.~Jakubowicz,
\newblock ``Convergence of a multi-agent projected stochastic gradient
  algorithm for non-convex optimization,''
\newblock {\em IEEE Transactions on Automatic Control}, vol. 58, no. 2, 2013.

\bibitem{Jakovetic2012}
Dusan Jakoveti\'{c}, José M.~F. Moura, and Xavier Joao,
\newblock ``{D}istributed {N}esterov-like gradient algorithms,''
\newblock in {\em Proc. 51st IEEE Conference on Decision and Control (CDC)}, 2012, pp. 5459--5464.

\bibitem{Duchi2012}
J.~C. Duchi, A.~Agarwal, and M.~J. Wainwright,
\newblock ``{D}ual {A}veraging for {D}istributed {O}ptimization: {C}onvergence
  {A}nalysis and {N}etwork {S}caling,''
\newblock {\em IEEE Transactions on Automatic Control}, vol. 57, no. 3, pp.
  592--606, 2012.

\bibitem{jadbabaie2009distributed}
Ali Jadbabaie, Asuman Ozdaglar, and Michael Zargham,
\newblock ``A distributed newton method for network optimization,''
\newblock in {\em Proc. 48th IEEE Conference on Decision and Control (CDC)}, 2009, pp. 2736--2741.

\bibitem{Lions1979}
P.L. Lions and B.~Mercier,
\newblock ``{S}plitting algorithms for the sum of two nonlinear operators,''
\newblock {\em SIAM Journal on Numerical Analysis}, vol. 16, no. 6, pp.
  964--979, 1979.

\bibitem{Eckstein1992}
J.~Eckstein and D.~P. Bertsekas,
\newblock ``{O}n the {D}ouglas-{R}achford splitting method and the proximal
  point algorithm for maximal monotone operators,''
\newblock {\em Mathematical Programming}, vol. 55, pp. 293--318, 1992.

\bibitem{combettes-pesquet-book1-2011}
Patrick~L. Combettes and Jean-Christophe Pesquet,
\newblock {\em Fixed-Point Algorithms for Inverse Problems in Science and
  Engineering}, chapter Proximal Splitting Methods in Signal Processing, pp.
  185--222,
\newblock Springer, 2011.

\bibitem{Boyd2011}
S.~Boyd, N.~Parikh, E.~Chu, B.~Peleato, and J.~Eckstein,
\newblock {\em {D}istributed optimization and statistical learning via the
  alternating direction method of multipliers}, vol.~3 of {\em Foundations and
  Trends in Machine Learning},
\newblock Now Publishers Inc., 2011.

\bibitem{Schizas2008}
I.D. Schizas, A.~Ribeiro, and G.B. Giannakis,
\newblock ``{C}onsensus in {A}d {H}oc {WSN}s {W}ith {N}oisy {L}inks -- {P}art
  {I}: {D}istributed {E}stimation of {D}eterministic {S}ignals,''
\newblock {\em Signal Processing, IEEE Transactions on}, vol. 56, no. 1, pp.
  350 --364, jan. 2008.

\bibitem{Mota2012}
Joao F.~C. Mota, Joao M.~F. Xavier, Pedro M.~Q. Aguiar, and Markus Puschel,
\newblock ``{D}istributed {ADMM} for model predictive control and congestion
  control,''
\newblock in {\em Proc. 51st IEEE Conference on Decision and Control (CDC)},
  2012, pp. 5110--5115.

\bibitem{Wei2012}
Ermin Wei and Asuman Ozdaglar,
\newblock ``{D}istributed {A}lternating {D}irection {M}ethod of
  {M}ultipliers,''
\newblock in {\em Proc. 51st IEEE Conference on Decision and Control (CDC)},
  2012, pp. 5445--5450.

\bibitem{Rockafellar1976}
R~Tyrrell Rockafellar,
\newblock ``{M}onotone operators and the proximal point algorithm,''
\newblock {\em SIAM Journal on Control and Optimization}, vol. 14, no. 5, pp.
  877--898, 1976.

\bibitem{Borwein2006}
J.M. Borwein and A.S. Lewis,
\newblock {\em {C}onvex {A}nalysis and {N}onlinear {O}ptimization : {T}heory
  and {E}xamples},
\newblock Springer Verlag, 2006.

\bibitem{Bauschke2011}
Heinz~H Bauschke and Patrick~L Combettes,
\newblock {\em {C}onvex analysis and monotone operator theory in {H}ilbert
  spaces},
\newblock Springer, 2011.

\bibitem{Boyd2006}
S.~Boyd, A.~Ghosh, B.~Prabhakar, and D.~Shah,
\newblock ``{R}andomized gossip algorithms,''
\newblock {\em IEEE Transactions on Information Theory}, vol. 52, no. 6, pp.
  2508--2530, 2006.

\bibitem{Tsitsiklis1986}
John~N. Tsitsiklis, Dimitri~P. Bertsekas, and Mlichael Athans,
\newblock ``{D}istributed asynchronous deterministic and stochastic gradient
  optimization algorithms,''
\newblock {\em IEEE Transactions on Automatic Control}, vol. 31, no. 9, pp.
  803--812, 1986.

\end{thebibliography}
\end{document}